\documentclass[journal]{IEEEtran}
\usepackage[centertags]{amsmath}
\usepackage{epsfig}
\usepackage{latexsym}
\usepackage{indentfirst}
\usepackage{mathrsfs}
\usepackage{amsthm}
\usepackage{amssymb}
\usepackage{amsfonts}
\usepackage{url}
\usepackage{pdflscape}
\usepackage{xcolor}
\usepackage{enumitem}
\newtheoremstyle{mytheorem}
  {}
  {}
  {\itshape}
  {}
  {}
  {.}
  { }
  {}

\theoremstyle{mytheorem}
\newtheorem{thm}{Theorem}
\newtheorem{cor}[thm]{Corollary}
\newtheorem{lem}[thm]{Lemma}
\newtheorem{prop}[thm]{Proposition}

\theoremstyle{remark}
\newtheorem{defn}[thm]{Definition}
\newtheorem{rem}[thm]{Remark}
\newtheorem{ex}[thm]{Example}

\newcommand{\Fq}{\mathbb{F}_{q}}

\newcommand{\F}{{\mathbb F}}

\newcommand{\K}{\mathbb{K}}

\newcommand{\ie}{\emph{i.e.}}

\newcommand{\sk}{\text{\hspace{1pt}-\hspace{1pt}}}
\newcommand{\se}{\text{\hspace{1pt}=\hspace{1pt}}}

\begin{document}
\title{Spectral Bounds for Quasi-Twisted Codes}

\author{Martianus Frederic Ezerman, San Ling, Buket \"{O}zkaya, and Jareena Tharnnukhroh
	\thanks{The authors are with the School of Physical and Mathematical Sciences, Nanyang Technological University, 21 Nanyang Link, Singapore 637371, e-mails: ${\rm\{fredezerman,\, lingsan,\,buketozkaya,\,jareena001\}}$@ntu.edu.sg.}
	\thanks{
	M.~F.~Ezerman, S.~Ling, and B.~\"{O}zkaya are supported by Nanyang Technological University Research Grant M4080456.}
	\thanks{J.~Tharnnukhroh's scholarship is from the Development and Promotion of Science and Technology (DPST) talent project of Thailand.}
	\thanks{This work has been accepted for presentation at the International Symposium on Information Theory ISIT 2019. Copyright (c) 2019 IEEE. Personal use of this material is permitted. However, permission to use this material for any other purposes must be obtained from the IEEE by sending a request to pubs-permissions@ieee.org.}
}

\maketitle

\begin{abstract}
New lower bounds on the minimum distance of quasi-twisted codes over finite fields are proposed. They are based on spectral analysis and eigenvalues of polynomial matrices. They generalize the Semenov-Trifonov and Zeh-Ling bounds in a manner similar to how the Roos and shift bounds extend the BCH and HT bounds for cyclic codes. 
\end{abstract}

\begin{IEEEkeywords}
Quasi-twisted code, Roos bound, shift bound, eigenvalues, polynomial matrices, spectral analysis.
\end{IEEEkeywords}

%
\IEEEpeerreviewmaketitle

\section{Introduction}
\label{intro}

Quasi-twisted (QT) codes form an important class of block codes that includes cyclic codes, quasi-cyclic (QC) codes and constacyclic codes as special subclasses. In addition to their rich algebraic structure (\cite{Y}), QT codes are also asymptotically good (\cite{C,DH}) and they yield good parameters (\cite{AGOSS1,AGOSS2}).

Several bounds on the minimum distance of cyclic codes had been derived. The first and perhaps the most famous one was the BCH bound, given by Bose and Chaudhuri (\cite{BC}), and by Hocquenghem (\cite{H}). An extension of the bound was formulated by Hartmann and Tzeng in \cite{HT}. One can consider the HT bound as a two-directional BCH bound. The Roos bound in \cite{R2} generalized this idea further by allowing the HT bound to have a certain number of gaps in both directions. The Roos bound was extended to constacyclic codes in \cite{RZ}. Another remarkable extension of the HT bound, known as the shift bound, was introduced by van Lint and Wilson in \cite{LW}. This bound is known to be particularly powerful on many non-binary codes (\cite{EL}). 

Despite being interesting from both theoretical and practical points of view, studies on the minimum distance estimates for QC and QT codes are not as rich as for cyclic and constacyclic codes. Semenov and Trifonov developed a spectral analysis of QC codes (\cite{ST}), based on the work done by Lally and Fitzpatrick in~\cite{LF}, and formulated a BCH-like bound, together with a comparison with a few other bounds for QC codes. Their approach is generalized by Zeh and Ling, by using the HT bound, in \cite{ZL}.

This paper is organized as follows. Section \ref{basics} recalls necessary background material and adapts the spectral method of Semenov-Trifonov to QT codes. We formulate and prove a generalized spectral bound on the minimum distance in Section \ref{bounds section}, where the Roos and shift bounds for QT codes are derived as special cases. Section \ref{exmp section} supplies numerical examples showing how the proposed bound performs in comparison with the Semenov-Trifonov (ST) and Zeh-Ling (ZL) bounds.
 
\section{Background}\label{basics}
\subsection{Constacyclic codes and minimum distance bounds from their defining sets}
Let $\Fq$ denote the finite field with $q$ elements, where $q$ is a prime power. Let $m$ be, throughout, a positive integer with $\gcd(m,q)=1$. For some nonzero element $\lambda \in \Fq$, a linear code $C\subseteq \Fq^m$ is called a {\it $\lambda$-constacyclic code} if it is invariant under the $\lambda$-constashift of codewords, \ie, $(c_0,\ldots,c_{m-1}) \in C$ implies $(\lambda c_{m-1},c_0,\ldots,c_{m-2}) \in C$. In particular, if $\lambda = 1$ or $q=2$, then $C$ is a cyclic code.

Consider the principal ideal $I=\langle x^m-\lambda \rangle$ of $\Fq[x]$ and define the residue class ring $R:=\Fq[x]/I$. To a vector $\vec{a}\in \Fq^m$, we associate an element of $R$ via the isomorphism:
\begin{eqnarray}\label{identification-1}
\phi: \ \F_q^{m} & \longrightarrow & R  \\
\vec{a}=(a_0,\ldots,a_{m-1}) & \longmapsto & a(x)= a_0+\cdots + a_{m-1}x^{m-1}.\nonumber
\end{eqnarray}
Note that the $\lambda$-constashift in $\F_q^{m}$ amounts to multiplication by $x$ in $R$. Hence, a $\lambda$-constacyclic code $C\subseteq \Fq^m$ can be viewed as an ideal of $R$. Since $R$ is a principal ideal ring, there exists a unique monic polynomial $g(x)\in R$ such that $C=\langle g(x)\rangle$, \ie, each codeword $c(x)\in C$ is of the form $c(x)=a(x)g(x)$, for some $a(x)\in R$. The polynomial $g(x)$, which is a divisor of $x^m-\lambda$, is called the {\it generator polynomial} of $C$. 

Let $\mbox{wt}(c)$ denote the number of nonzero coefficients in $c(x)\in C$. Recall that the minimum distance of $C$ is defined as $d(C):=\min\{\mbox{wt}(c) : 0\neq c(x)\in C\}$ when $C$ is not the trivial zero code. For any positive integer $p$, let $\vec{0}_p$ denote throughout the all-zero vector of length $p$. A $\lambda$-constacyclic code $C=\{\vec{0}_m\}$ if and only if $g(x)=x^m-\lambda$. 

The roots of $x^m-\lambda$ are of the form $\alpha, \alpha\xi, \ldots, \alpha\xi^{m-1}$, where  $\alpha$ is a fixed $m^{th}$ root of $\lambda$ and $\xi$ is a fixed primitive $m^{th}$ root of unity. Henceforth, let $\Omega :=\{\alpha\xi^k\ :\ 0\leq k \leq m-1\}$ be the set of all $m^{th}$ roots of $\lambda$ and let $\F$ be the smallest extension of $\Fq$ that contains $\Omega$ (equivalently, $\F$ is the splitting field of $x^m-\lambda$). Given the $\lambda$-constacyclic code $C=\langle g(x)\rangle$, the set $L:=\{\alpha\xi^k\ :\ g(\alpha\xi^k)=0\}\subseteq \Omega$ of roots of its generator polynomial is called the {\it defining set} of $C$. Note that $\alpha\xi^k\in L$ implies $\alpha\xi^{qk}\in L$, for each $k$. A nonempty subset $E\subseteq\Omega$ is said to be {\it consecutive} if there exist integers $e,n$ and $\delta$ with $e\geq 0,\delta \geq 2, n> 0$ and $\gcd(m,n)=1$ such that
\begin{equation} \label{cons zero set}
E:=\{\alpha\xi^{e+zn}\ :\ 0\leq z\leq \delta-2\}\subseteq\Omega.
\end{equation}

We now describe the Roos bound for constacyclic codes (see \cite[Theorem 2]{R2} for the original Roos bound for cyclic codes). For $P\subseteq\Omega$, let $C_P$ denote any $\lambda$-constacyclic code of length $m$ over $\F_q$, whose defining set contains $P$. Let $d_P$ denote the minimum distance of $C_P$. 

\begin{thm}~\cite[Theorem 6]{RZ} (Roos bound) \label{Roos} Let $N$ and $M$ be two nonempty subsets of $\Omega$. If there exists a consecutive set $M'$ containing $M$ such that $|M'| \leq |M| + d_N -2$, then we have $d_{MN}\geq |M| + d_N -1$, where 
$\displaystyle{
MN:=\frac{1}{\alpha}\bigcup_{\varepsilon\in M} \varepsilon N}$.
\end{thm}
If $N$ is consecutive like in (\ref{cons zero set}), then we get the following.
\begin{cor}\cite[Corollary 1]{RZ}, \cite[Corollary 1]{R2} \label{Roos2}
Let $N, M$ and $M'$ be as in Theorem \ref{Roos}, with $N$ consecutive. Then $|M'| < |M| + |N|$ implies $d_{MN}\geq |M| + |N|$.
\end{cor}

\begin{rem}\label{Roos remark}
In particular, the case $M=\{\alpha\}$ yields the BCH bound for the associated constacyclic code (see \cite[Corollary 2]{RZ} and the original BCH bound for cyclic codes in \cite{BC} and \cite{H}). Taking $M'=M$ yields the HT bound (see \cite[Corollary 3]{RZ} and the HT bound for cyclic codes in \cite[Theorem 2]{HT}). 
\end{rem}

Another improvement to the HT bound for cyclic codes was given by van Lint and Wilson in \cite{LW}, which is known as the shift bound. We now formulate the shift bound for constacyclic codes. To do this, we need the notion of an {\it independent set}, which can be constructed over any field in a recursive way.

Let $S$ be a subset of some field $\K$ of any characteristic. One inductively defines a family of finite subsets of $\K$, called independent with respect to $S$, as follows.
\begin{enumerate}
\item $\emptyset$ is independent with respect to $S$.
\item If $A \subseteq S$ is independent with respect to $S$, then $A\cup\{b\}$ is independent with respect to $S$ for all $b\in \K \setminus S$. \label{cond:2}
\item If $A$ is independent with respect to $S$ and $c\in\K^{*}$, then $cA$ is independent with respect to $S$. \label{cond:3}
\end{enumerate}

Recall that the {\it weight} of a polynomial $f(x) \in \K[x]$, denoted by $\mbox{wt}(f)$, is the number of nonzero coefficients in $f(x)$.

\begin{thm} \cite[Theorem 11]{LW} (Shift bound)\label{shift bound}
Let $0\neq f(x)\in \K[x]$ and let $S:=\{\theta\in \K\ | \ f(\theta)=0\}$. Then $\mbox{wt}(f)\geq |A|,$ for every subset $A$ of $\K$ that is independent with respect to $S$.
\end{thm}

The minimum distance bound for a given $\lambda$-constacyclic code follows by considering the weights of its codewords $c(x)\in C$ and the independent sets with respect to subsets of its defining set $L$. Observe that, in this case, the universe of the independent sets is $\Omega$, not $\F$, because all of the possible roots of the codewords are contained in $\Omega$. Moreover, we choose $b$ from $\Omega \setminus P$ in Condition \ref{cond:2}) above, where $P\subseteq L$, and $c$ in Condition \ref{cond:3}) is of the form $\xi^k\in\F^{*}$, for some $0\leq k\leq m-1$. 

\begin{rem}\label{shift remark}
In particular, $A=\{\alpha\xi^{e+zn}\ :\ 0\leq z\leq \delta-1\}$ is independent with respect to the consecutive set $E$ in (\ref{cons zero set}), which gives the BCH bound for $C_E$. Let $D := \{\alpha\xi^{e + z  n_1 + y  n_2} : 0\leq z \leq \delta-2, 0\leq y\leq s\}$, for integers $b \geq 0$, $\delta \geq 2$ and positive integers $s, n_1$ and $n_2$ such that $\gcd(m, n_1) = 1$ and $\gcd(m, n_2) < \delta$. Then, for any fixed $z\in \{0,\ldots,\delta-2\}$, $A_z:=\{\alpha\xi^{e + z  n_1} :	0\leq z \leq\delta-2\}\cup\{\alpha\xi^{e + z  n_1 + y  n_2} : 0\leq y \leq s+1\}$ is independent with respect to $D$ and we get the HT bound for $C_D$.
\end{rem}

\subsection{Spectral theory of quasi-twisted codes}
A linear code $C\subseteq\Fq^{m\ell}$ is called {\it $\lambda$-quasi-twisted} ($\lambda$-QT) of index $\ell$ if it is invariant under the $\lambda$-constashift of codewords by $\ell$ positions with $\ell$ being the smallest positive integer with this property. In particular, if $\ell=1$, then $C$ is $\lambda$-constacyclic. If $\lambda = 1$ or $q=2$, then $C$ is QC of index $\ell$. For a codeword $\vec{c}\in C$, seen as an $m \times \ell$ array
\begin{equation}\label{array}
\vec{c}=\left(
  \begin{array}{ccc}
    c_{00} & \ldots & c_{0,\ell-1} \\
    \vdots & \vdots  & \vdots \\
    c_{m-1,0} & \ldots & c_{m-1,\ell-1} \\
  \end{array}
\right),
\end{equation} 
being invariant under $\lambda$-constashift by $\ell$ units in $\Fq^{m\ell}$ corresponds to being closed under row $\lambda$-constashift in $\Fq^{m\times\ell}$.

To an element $\vec{c}\in \Fq^{m\times \ell} \simeq \Fq^{m\ell}$ in (\ref{array}), we associate an element of $R^\ell$ (cf. (\ref{identification-1}))
\begin{equation} \label{associate-1}
\vec{c}(x):=(c_0(x),c_1(x),\ldots ,c_{\ell-1}(x)) \in R^\ell ,
\end{equation}
where, for each $0\leq j \leq \ell-1$, 
\begin{equation*}
c_j(x):= c_{0,j}+c_{1,j}x+c_{2,j}x^2+\cdots + c_{m-1,j}x^{m-1} \in R .
\end{equation*} 
The isomorphism $\phi$ in (\ref{identification-1}) extends naturally to
\begin{equation}\begin{array}{rll} \label{identification-2}
\Phi: \F_q^{m\ell} & \longrightarrow & R^\ell  \\
\vec{c} \hspace{5pt}& \longmapsto & \vec{c}(x) .
\end{array}\end{equation}
The  row $\lambda$-constashift in $\Fq^{m\times\ell}$ corresponds to componentwise multiplication by $x$ in $R^\ell$. The map $\Phi$ above is, therefore, an $R$-module isomorphism and a $\lambda$-QT code $C\subseteq \F_q^{m\ell}$ of index $\ell$ can be viewed as an $R$-submodule of $R^\ell$.

Lally and Fitzpatrick proved in~\cite{LF} that every QC code has a polynomial generator in the form of a reduced matrix. We provide an easy adaptation of their findings for QT codes. 

Consider the ring homomorphism
\begin{eqnarray}\label{eq:Psi}
\Psi \ :\ \Fq[x]^{\ell} &\longrightarrow& R^{\ell} \\\nonumber
(f_0(x),\ldots, f_{\ell-1}(x))  &\longmapsto& (f_0(x)+I,\ldots ,f_{\ell-1}(x)+I).
\end{eqnarray}
Let each $\vec{e}_j$ denote the standard basis vector of length $\ell$ with $1$ at the $j^{th}$ coordinate and $0$ elsewhere. Given a $\lambda$-QT code $C\subseteq R^{\ell}$, the preimage $\widetilde{C}$ of $C$ in $\Fq[x]^{\ell}$ is an $\Fq[x]$-submodule containing $\widetilde{K} =\{(x^m-\lambda)\vec{e}_j : 0\leq j \leq \ell-1\}$. From here on, the tilde indicates structures over $\Fq[x]$.

Since $\widetilde{C}$ is a submodule of the finitely generated free module $F_q[x]^{\ell}$ over the principal ideal domain $\Fq[x]$ and contains $\widetilde{K}$, it has a generating set of the form $$\{\vec{u}_1,\ldots,\vec{u}_p, (x^m-\lambda)\vec{e}_0,\ldots,(x^m-\lambda)\vec{e}_{\ell-1}\},$$  where $p\geq 1$ is an integer and $\vec{u}_b = (u_{b,0}(x), \ldots, u_{b,\ell-1}(x)) \in \Fq[x]^{\ell}$, for each $b \in \{1,\ldots,p\}$. Hence, the rows of 

$$\mathcal{G}=\left(\begin{array}{ccc}
    u_{1,0}(x) & \ldots & u_{1,\ell-1}(x) \\
    \vdots & \vdots & \vdots \\
    u_{p,0}(x) & \ldots & u_{p,\ell-1}(x) \\
     x^m-\lambda & \ldots & 0 \\
    \vdots & \ddots & \vdots \\
    0 & \ldots & x^m-\lambda \\
  \end{array}
\right)$$
generate $\widetilde{C}$. We triangularise $\mathcal{G}$ by elementary row operations to obtain another equivalent generating set from the rows of an upper-triangular $\ell \times \ell$ matrix with entries in $\Fq[x]$ 
\begin{equation}\label{PGM}
\widetilde{G}(x)=\left(\begin{array}{cccc}
    g_{0,0}(x) & g_{0,1}(x) & \ldots & g_{0,\ell-1}(x) \\
    0 & g_{1,1}(x) & \ldots & g_{1,\ell-1}(x) \\
    \vdots & \vdots & \ddots & \vdots \\
    0 & 0 &\ldots & g_{\ell-1,\ell-1}(x)\\
  \end{array}
\right),
\end{equation}
where $\widetilde{G}(x)$ satisfies (see \cite[Theorem 2.1]{LF}):
\begin{enumerate}
    \item $g_{i,j}(x)=0$ for all $0\leq j < i \leq \ell-1$.
    \item $\deg(g_{i,j}(x)) < $ $\deg(g_{j,j}(x))$ for all $i <j$. \label{cond:2nd}
    \item $g_{j,j}(x) \mid (x^m-\lambda)$ for all $0\leq j \leq \ell-1$.
    \item If  $g_{j,j}(x) = (x^m-\lambda)$, then $g_{i,j}(x) =0$ for all $i\neq j$.
\end{enumerate}
Note that $\widetilde{G}(x)$ has nonzero rows and each nonzero element of $\widetilde{C}$ can be expressed as $(0,\ldots,0,c_j(x),\ldots,c_{\ell-1}(x))$, where $j\geq 0$, $ c_j(x)\neq 0$ and $g_{j,j}(x)\mid c_j(x)$. Moreover, Condition \ref{cond:2nd}) implies that the rows of $\widetilde{G}(x)$ is a reduced basis of $\widetilde{C}$, which is uniquely defined, up to multiplication by constants, with monic diagonal elements.

Let $G(x)$ be the matrix with the rows of $\widetilde{G}(x)$ under the image of $\Psi$ in (\ref{eq:Psi}). Then, the rows of $G(x)$ are an $R$-generating set for $C$. We say that $C$, generated as an $R$-submodule, is an $r$-generator QT code if $G(x)$ has $r$ (nonzero) rows. The $\Fq$-dimension of $C$, as shown in~\cite[Corollary 2.4]{LF}, is 
\begin{equation}\label{dimension}
m\ell-\sum_{j=0}^{\ell-1} \deg(g_{j,j}(x))=\sum_{j=0}^{\ell-1}\left[m -\deg(g_{j,j}(x))\right].
\end{equation}

In \cite{ST}, Semenov and Trifonov use the polynomial matrix $\widetilde{G}(x)$ in (\ref{PGM}) to develop a spectral theory for QC codes. This gives rise to a BCH-like minimum distance bound. Their bound is improved by Zeh and Ling in \cite{ZL} by using the HT bound (\cite{HT}). We translate their results from QC to QT codes. 

Given a $\lambda$-QT code $C\subseteq R^{\ell}$, let the associated $\ell \times \ell$ upper triangular matrix $\widetilde{G}(x)$ be as in (\ref{PGM}) with entries in $\Fq[x]$. The {\it determinant} of $\widetilde{G}(x)$ is $$\det(\widetilde{G}(x)):=\prod_{j=0}^{\ell-1}g_{j,j}(x)$$
and an {\it eigenvalue} $\beta$ of $C$ is a root of $\det(\widetilde{G}(x))$. Note that, since $g_{j,j}(x)\mid x^m-\lambda$, for each $0\leq j\leq \ell-1$, all eigenvalues are elements of $\Omega$, \ie, $\beta=\alpha\xi^k$ for some $k\in\{0,\ldots,m-1\}$. The {\it algebraic multiplicity} of $\beta$ is the largest integer $a$ such that $(x-\beta)^a\mid \det(\widetilde{G}(x))$. The {\it geometric multiplicity} of $\beta$ is the dimension of the null space of $\widetilde{G}(\beta)$. This null space, denoted by $\mathcal{V}_{\beta}$, is called the {\it eigenspace} of $\beta$. In other words, 
\[
\mathcal{V}_{\beta}:=\big\{\vec{v}\in\F^{\ell} : \widetilde{G}(\beta)\vec{v}^{\top}=\vec{0}_{\ell}\big\},
\]
where $\F$ is the splitting field of $x^m-\lambda$, as before. It was shown in \cite{ST} that, for a given QC code and the associated $\widetilde{G}(x) \in (\F_q[x])^{\ell\times\ell}$, the algebraic multiplicity $a$ of an eigenvalue $\beta$ is equal to its geometric multiplicity $\dim_{\F}(\mathcal{V}_{\beta})$. We state the QT analogue of this result without the proof, since it can be shown in exactly the same way.
\begin{lem}\cite[Lemma 1]{ST}\label{multiplicity lemma}
The algebraic multiplicity of any eigenvalue of a $\lambda$-QT code $C$ is equal to its geometric multiplicity.
\end{lem}

From this point on, we let $\overline{\Omega}\subseteq \Omega$ denote the nonempty set of all eigenvalues of $C$ such that $|\overline{\Omega}|=t>0$. Note that $\overline{\Omega}=\emptyset$ if and only if the diagonal elements $g_{j,j}(x)$ in $\widetilde{G}(x)$ are constant and $C$ is the trivial full space code. Choose an arbitrary eigenvalue $\beta_i\in\overline{\Omega}$ with multiplicity $n_i$ for some $i \in\{1,\ldots,t\}$. Let $\{\vec{v}_{i,0},\ldots,\vec{v}_{i,n_i-1}\}$ be a basis for the corresponding eigenspace $\mathcal{V}_i$. Consider the matrix  
\begin{equation}\label{Eigenspace} 
 V_i:=\begin{pmatrix}
\vec{v}_{i,0} \\
\vdots\\
\vec{v}_{i,n_i-1} 
\end{pmatrix}
=
\begin{pmatrix}
v_{i,0,0}&\ldots&v_{i,0,\ell-1} \\
\vdots & \vdots & \vdots\\
v_{i,n_i-1,0}&\ldots&v_{i,n_i-1,\ell-1}
 \end{pmatrix},
\end{equation} 
having the basis elements as its rows. We let
\[
H_i:=(1, \beta_i,\ldots,\beta_i^{m-1})\otimes V_i \mbox{ and }
\]
\begin{equation}\label{parity check matrix} 
H:=\begin{pmatrix}
H_1 \\
\vdots\\
H_t 
\end{pmatrix}=\begin{pmatrix}
V_1&\beta_1 V_1 & \ldots &(\beta_1)^{m-1} V_1 \\
\vdots &\vdots & \vdots & \vdots\\
V_t & \beta_t V_t & \ldots & (\beta_t)^{m-1} V_t
\end{pmatrix}.
\end{equation} 
Observe that $H$ has $n:=\sum_{i=1}^t n_i$ rows. By Lemma \ref{multiplicity lemma}, we have $n=\sum_{j=0}^{\ell-1}\mbox{deg}(g_{j,j}(x))$. To prove Lemma~\ref{rank lemma} below, it remains to show the linear independence of these $n$ rows, which was already shown in~\cite[Lemma 2]{ST}.
\begin{lem}\label{rank lemma}
The matrix $H$ in (\ref{parity check matrix}) has rank $m\ell -\dim_{\Fq}(C)$.
\end{lem}

It is immediate to confirm that $H \vec{c}^{\top}=\vec{0}_n$ for any codeword $\vec{c}\in C$. Together with Lemma \ref{rank lemma}, we obtain the following easily.

\begin{prop}\cite[Theorem 1]{ST}
The $n\times m\ell$ matrix $H$ in (\ref{parity check matrix}) is a parity-check matrix for $C$.
\end{prop}

\begin{rem}
Note that if $\overline{\Omega}=\emptyset$, then the construction of $H$ in (\ref{parity check matrix}) is impossible. Hence, we have assumed $\overline{\Omega}\neq\emptyset$ and we can always say $H=\vec{0}_{m\ell}$ if $C=\Fq^{m\ell}$. The other extreme case is when $\overline{\Omega}=\Omega$. By using Lemma \ref{rank lemma} above, one can easily deduce that a given QT code $C=\{\vec{0}_{m\ell}\}$ if and only if $\overline{\Omega}=\Omega$, each $\mathcal{V}_i=\F^{\ell}$ (equivalently, each $V_i=I_{\ell}$, where $I_{\ell}$ denotes the $\ell\times\ell$ identity matrix) and $n=m\ell$ so that we obtain $H=I_{m\ell}$. On the other hand, $\overline{\Omega}=\Omega$ whenever $x^m-\lambda \mid \mbox{det}(\widetilde{G}(x))$, but $C$ is nontrivial unless each eigenvalue in $\Omega$ has multiplicity $\ell$.
\end{rem}

\begin{defn}\label{eigencode}
Let $\mathcal{V}\subseteq \F^\ell$ be an eigenspace. We define the {\it eigencode} corresponding to $\mathcal{V}$ by
$$\mathbb{C}(\mathcal{V})=\mathbb{C}:=\left\{\vec{u}\in \Fq^\ell\ : \ \sum_{j=0}^{\ell-1}{v_ju_j}=0, \forall \vec{v} \in \mathcal{V}\right\}.$$
In case we have $\mathbb{C}=\{\vec{0}_{\ell}\}$, then it is assumed that $d(\mathbb{C})=\infty$.
\end{defn}

The BCH-like minimum distance bound of Semenov and Trifonov for a given QC code in ~\cite[Theorem 2]{ST} is expressed in terms of the size of a consecutive subset of eigenvalues in $\overline{\Omega}$ and the minimum distance of the common eigencode related to this consecutive subset. Zeh and Ling generalized their approach and derived an HT-like bound in~\cite[Theorem 1]{ZL} without using the parity-check matrix in their proof. The eigencode, however, is still needed. In the next section we will prove the analogues of these bounds for QT codes in terms of the Roos and shift bounds. 

\section{Spectral Bounds for QT Codes} \label{bounds section}
First, we establish a general spectral bound on the minimum distance of a given QT code. Let $C\subseteq\Fq^{m\ell}$ be a $\lambda$-QT code of index $\ell$ with nonempty eigenvalue set $\overline{\Omega}\varsubsetneqq\Omega$. Let $P\subseteq\overline{\Omega}$ be a nonempty subset of eigenvalues such that $P=\{\alpha\xi^{u_1}, \alpha\xi^{u_2},\ldots,\alpha\xi^{u_r}\}$, where $0<r\leq|\overline{\Omega}|$.  We define 
\begin{equation}\label{pmatrix}
\widetilde{H}_P:=\begin{pmatrix}
1&\alpha\xi^{u_1}&(\alpha\xi^{u_1})^2&\ldots&(\alpha\xi^{u_1})^{m-1}\\
\vdots & \vdots & \vdots & \vdots & \vdots \\
1&\alpha\xi^{u_r}&(\alpha\xi^{u_r})^2&\ldots&(\alpha\xi^{u_r})^{m-1}
\end{pmatrix}.
\end{equation}
Let $d_P$ be a nonnegative integer such that any $\lambda$-constacyclic code $C_P\subseteq\F_q^m$, whose defining set contains $P$, has a minimum distance at least $d_P$. We have $\widetilde{H}_P\vec{c}_P^{\top}=\vec{0}_r$, for any $\vec{c}_P\in C_P$. In particular, if $P$ is equal to the defining set of $C_P$, then $\widetilde{H}_P$ is a parity-check matrix for $C_P$.

Let $\mathcal{V}_P$ denote the common eigenspace of the eigenvalues in $P$ and let $V_P$ be the matrix, say of size $t\times\ell$, consisting of a basis for $\mathcal{V}_P$ (cf. (\ref{Eigenspace})). If we set $\widehat{H}_P =\widetilde{H}_P \otimes V_P $, then $\widehat{H}_P \vec{c}^{\top}=\vec{0}_{m\ell}$, for all $\vec{c}\in C$. In other words, $\widehat{H}_P$ is a submatrix of $H$ in (\ref{parity check matrix}) if $\mathcal{V}_P\neq\{\vec{0}_{\ell}\}$. If $\mathcal{V}_P=\{\vec{0}_{\ell}\}$, then $\widehat{H}_P$ does not exist. We first handle this case separately so that the bound is valid even if we have $\mathcal{V}_P=\{\vec{0}_{\ell}\}$, before the cases where we can use $\widehat{H}_P$ in the proof.

In the rest, we consider the quantity $\min(d_P, d(\mathbb{C}_P))$, where $\mathbb{C}_P$ is the eigencode corresponding to $\mathcal{V}_P$. We have assumed $P\neq\emptyset$ so that $\widetilde{H}_P$ is defined, and we also have $P\neq\Omega$ as $P\subseteq\overline{\Omega}\varsubsetneqq\Omega$ so that $d_P$ is well-defined. If $|P|\geq 1$, then the BCH bound implies $d_P\geq 2$. On the other hand, if $\mathcal{V}_P=\{\vec{0}_{\ell}\}$, then $\mathbb{C}_P=\Fq^{\ell}$ and $d(\mathbb{C}_P)=1$. Hence, $\min(d_P, d(\mathbb{C}_P))=1$ only if $d(\mathbb{C}_P)=1$ (including the case $\mathcal{V}_P=\{\vec{0}_{\ell}\}$), where $d(C)\geq 1$ holds for any nonzero QT code $C$.

Now let $\emptyset\neq P\subseteq\overline{\Omega}\varsubsetneqq\Omega$ and $d(\mathbb{C}_P)\geq 2$. Assume that there exists a codeword $\vec{c}\in C$ of weight $\omega$ such that $0 < \omega < \min(d_P, d(\mathbb{C}_P))$. For each $0 \leq k \leq m-1$, let $\vec{c}_k = (c_{k,0}, . . . , c_{k,\ell-1})$ be the $k^{th}$ row of the codeword $\vec{c}$ given as in (\ref{array}) and we set $\vec{s}_k := V_P\vec{c}_k^{\top}$. Since $d(\mathbb{C}_P) > \omega$, we have $\vec{c}_k \notin \mathbb{C}_P$ and therefore $ \vec{s}_k=V_P\vec{c}_k^{\top}\neq \vec{0}_t$, for all $\vec{c}_k \neq \vec{0}_{\ell}$, $k\in\{0,\ldots,m-1\}$. Hence, $0 < \lvert\{\vec{s}_k : \vec{s}_k \neq \vec{0}_t \} \rvert \leq \omega < \min(d_P, d(\mathbb{C}_P))$. Let $S := [\vec{s}_0\ \vec{s}_1 \cdots \vec{s}_{m-1}]$. Then $\widetilde{H}_PS^{\top} = 0$, which implies that the rows of the matrix $S$ lies in the right kernel of $\widetilde{H}_P$. But this is a contradiction since any row of $S$ has weight at most $\omega<d_P$, showing the following.

\begin{thm}\label{main thm}
Let $C\subseteq R^\ell$ be a $\lambda$-QT code of index $\ell$ with nonempty eigenvalue set $\overline{\Omega}\varsubsetneqq\Omega$. Let $P\subseteq\overline{\Omega}$ be a nonempty subset of eigenvalues and let $C_P\subseteq\F_q^m$ be any $\lambda$-constacyclic code with defining set $L\supseteq P$ and minimum distance at least $d_P$. We define $\mathcal{V}_P:=\bigcap_{\beta\in P}\mathcal{V}_{\beta}$ as the common eigenspace of the eigenvalues in $P$ and let $\mathbb{C}_P$ denote the eigencode corresponding to $\mathcal{V}_P$. Then, 
\begin{equation}\label{gen spectral}
d(C) \geq \min \left\{ d_P, d(\mathbb{C}_P) \right\}.
\end{equation}
\end{thm}

Theorem \ref{main thm} allows us to use any minimum distance bound derived for constacyclic codes based on their defining set. The following special cases are immediate after the preparation that we have done in Section \ref{basics} (cf. Theorems \ref{Roos} and \ref{shift bound}).

\begin{cor}\label{Cor-Roos-Shift}
Let $C\subseteq R^\ell$ be a $\lambda$-QT code of index $\ell$ with $\overline{\Omega}\varsubsetneqq\Omega$ as its nonempty set of eigenvalues.
\begin{enumerate}[leftmargin=*]
\item[i.] Let $N$ and $M$ be two nonempty subsets of $\Omega$ such that $MN\subseteq\overline{\Omega}$, where $MN:=\frac{1}{\alpha}\bigcup_{\varepsilon\in M} \varepsilon N$. If there exists a consecutive set $M'$ containing $M$ with $|M'|\leq |M|+d_N-2$, then $d(C)\geq \min(|M|+d_N-1,d(\mathbb{C}_{MN}))$.\vspace{3pt}

\item[ii.] For every $A\subseteq\Omega$ that is independent with respect to $\overline{\Omega}$, we have $d(C)\geq \min(|A|,d(\mathbb{C}_{T_A}))$, where $T_A:=A\cap \overline{\Omega}$.
\end{enumerate}
\end{cor}

\begin{proof}\hfill
\begin{enumerate}[leftmargin=*]
\item[i)] Let $N=\{\alpha\xi^{u_1},\ldots,\alpha\xi^{u_r}\}$ and $M=\{\alpha\xi^{v_1},\ldots,\alpha\xi^{v_s}\}$ be such that there exists a consecutive set $M'=\{\alpha\xi^z: v_1\leq z\leq v_s\}\subseteq\Omega$ containing $M$ with $|M'|\leq |M|+d_N-2$. We define the matrices 
\small{\[\widetilde{H}_{N}:=\setlength\arraycolsep{4pt}\begin{pmatrix}
1&\alpha\xi^{u_1}&(\alpha\xi^{u_1})^2&\ldots&(\alpha\xi^{u_1})^{m-1}\\
\vdots & \vdots & \vdots & \vdots & \vdots \\
1&\alpha\xi^{u_r}&(\alpha\xi^{u_r})^2&\ldots&(\alpha\xi^{u_r})^{m-1}
\end{pmatrix},\]\vspace{3pt}
\[\widetilde{H}_{M}:=\setlength\arraycolsep{4pt}\begin{pmatrix}
1&\alpha\xi^{v_1}&(\alpha\xi^{v_1})^2&\ldots&(\alpha\xi^{v_1})^{m-1}\\
\vdots & \vdots & \vdots & \vdots & \vdots \\
1&\alpha\xi^{v_s}&(\alpha\xi^{v_s})^2&\ldots&(\alpha\xi^{v_s})^{m-1}
\end{pmatrix}. \]}
\normalsize
Consider the joint subset $MN=\{\alpha\xi^{u_i+v_j} : 1\leq i\leq r, 1\leq j\leq s\} \subseteq\overline{\Omega}$. Let $B_k$ be the $k^{th}$ column of $\widetilde{H}_N$ for $k \in \{0,\ldots,m-1\}$. We create the joint matrix
\small{\begin{equation*}
\widetilde{H}_{MN}=\setlength\arraycolsep{3pt}\begin{pmatrix}
B_0 &\alpha^{v_1}B_1&(\alpha^{v_1})^2B_2&\ldots&(\alpha^{v_1})^{m-1}B_{m-1}\\
\vdots & \vdots & \vdots & \vdots & \vdots \\
B_0 &\alpha^{v_s}B_1&(\alpha^{v_s})^2B_2&\ldots&(\alpha^{v_s})^{m-1}B_{m-1}
\end{pmatrix}.
\end{equation*}}
\normalsize
Now let $\mathcal{V}_{MN}:=\bigcap_{\beta\in MN}\mathcal{V}_{\beta}$ denote the common eigenspace of the eigenvalues in $MN$ and let $V_{MN}$ be the matrix consisting of a basis for $\mathcal{V}_{MN}$, built as in (\ref{Eigenspace}). Let $\mathbb{C}_{MN}$ be the eigencode corresponding to $\mathcal{V}_{MN}$.  Setting $\widehat{H}_{MN} :=\widetilde{H}_{MN} \otimes V_{MN}$ implies $\widehat{H}_{MN} \vec{c}^{\top}=\vec{0}$ for all $\vec{c}\in C$. The rest of the proof is identical with the proof of Theorem \ref{main thm}, where $P$ is replaced by $MN$, and the result follows by the Roos bound (Theorem \ref{Roos}).

\item[ii)] For each independent $A\subseteq\Omega$ with respect to $\overline{\Omega}$, let $T_A=A\cap \overline{\Omega}=\{\alpha\xi^{w_1},\alpha\xi^{w_2},\ldots,\alpha\xi^{w_y}\}$. Since $\overline{\Omega}$ is a proper subset of $\Omega$, a nonempty $T_A$ can be obtained by the recursive construction of $A$. We define 
$$\widetilde{H}_{T_A}=\begin{pmatrix}
1&\alpha\xi^{w_1}&(\alpha\xi^{w_1})^2&\ldots&(\alpha\xi^{w_1})^{m-1}\\
\vdots & \vdots & \vdots & \vdots  & \vdots \\
1&\alpha\xi^{w_y}&(\alpha\xi^{w_y})^2&\ldots&(\alpha\xi^{w_y})^{m-1}
\end{pmatrix}.$$
Let $V_{T_A}$ be the matrix corresponding to a basis of $\mathcal{V}_{T_A}$, which is the intersection of the eigenspaces belonging to the eigenvalues in $T_A$. Let $\mathbb{C}_{T_A}$ be the eigencode corresponding to the eigenspace $\mathcal{V}_{T_A}$. We again set $\widehat{H}_{T_A} := \widetilde{H}_{T_A} \otimes V_{T_A}$ and the result follows in a similar way by using the shift bound (Theorem \ref{shift bound}).
\end{enumerate}
\end{proof}

\begin{rem}
We can obtain the QT analogues of the BCH-like bound in \cite[Theorem 2]{ST} and the HT-like bound in \cite[Theorem 1]{ZL} by using Remarks \ref{Roos remark} and \ref{shift remark}.
\end{rem}

\section{Examples} \label{exmp section}

We begin with two examples of QC codes ($\lambda=1$) for which Corollary \ref{Cor-Roos-Shift} yields the exact distances. 

\begin{ex}
Let $\gamma$ be a primitive $23^{rd}$ root of unity. Let $C\subseteq\F_2^{92}$ be the binary QC code with $\ell=4$, $d(C)=7$ and eigenvalues $\overline{\Omega}=\{\gamma^i : i = 1, 2, 3, 4, 6, 8, 9, 12, 13, 16, 18\}$. The common eigenspace is generated by $V_{\overline{\Omega}}=V=I_4$ over $\F_{2^{11}}$, which is the splitting field of $x^{23}-1$. We have $\mathbb{C}_{\overline{\Omega}}=\{\vec{0}_4\}$ and therefore $d(\mathbb{C}_{\overline{\Omega}})=\infty$. Hence, Theorem \ref{main thm} yields $d(C)\geq d_P$, for $P=\overline{\Omega}$. The associated cyclic code $C_{\overline{\Omega}}$ is the well-known binary Golay code of length 23, which has minimum distance $d_{\overline{\Omega}}=7$, which is equal to $d(C)$. Note that the shift bound yields the exact distance of binary Golay code (see \cite[Example 7]{LW}) with $A=\{\gamma^i : i = 0, 1, 3, 4, 5, 6, 16, 18\}$ and $\mathbb{C}_{T_A}=\mathbb{C}_{\overline{\Omega}}=\{\vec{0}_4\}$, hence Corollary \ref{Cor-Roos-Shift} ii. is sharp in this example. We also note that the Roos bound estimates 5 for $d(C_{\overline{\Omega}})$, as does the BCH bound.
\end{ex}

\begin{ex}
Let $\eta$ be a primitive $26^{th}$ root of unity. Consider the ternary QC code $C\subseteq\F_3^{104}$ with $\ell=4$, minimum distance $6$ and eigenvalues $\overline{\Omega}=\{\eta^i : i = 0, 13, 14, 16, 17, 22, 23, 25\}$. The common eigenspace is generated by $V_{\overline{\Omega}}=V=I_4$ over $\F_{3^3}$, which is the splitting field of $x^{26}-1$. We again have $\mathbb{C}_{\overline{\Omega}}=\{\vec{0}_4\}$ and therefore $d(\mathbb{C}_{\overline{\Omega}})=\infty$.  Hence, Theorem \ref{main thm} yields $d(C)\geq d_{\overline{\Omega}}$. The cyclic code $C_{\overline{\Omega}}$ has minimum distance $d_{\overline{\Omega}}=6=d(C)$. Note that the Roos bound yields the exact distance of $C_{\overline{\Omega}}$: let $N=\{\eta^{13},\eta^{14}\}$ and $M=\{\eta^0,\eta^3,\eta^9,\eta^{12}\}$. Then $d_N=3$ and $M'=\{\eta^0,\eta^3,\eta^6,\eta^9,\eta^{12}\}=\{\eta^{3i} : 0\leq i\leq 4\}$, so $|M'|=5\leq 4+3-2$. We get $d_{MN}\geq 4+3-1$, where $\mathbb{C}_{MN}=\{\vec{0}_4\}$. However, the shift bound estimates 5 for $d(C_{\overline{\Omega}})$ (see \cite[Example 26.7]{EL}), hence Corollary \ref{Cor-Roos-Shift} i. is sharp here.
\end{ex}

In \cite{ST}, Semenov and Trifonov compared their BCH-like spectral bound with several other bounds given for QC codes. In Table~\ref{tab:comparison}, we compare the estimates of the general spectral bound given in (\ref{gen spectral}) with the ST and ZL bounds for a number of binary and ternary codes. The actual distance of the QT code and the estimates of the spectral, ST and ZL bounds are denoted by $d, d_{SP}, d_{BCH}$ and $d_{HT}$, respectively. We consider the case $P=\overline{\Omega}$ so that $d_{\overline{\Omega}}=d(C_{\overline{\Omega}})$, and the search using Magma (\cite{BCP}) is restricted to QT codes with $\mathbb{C}_{\overline{\Omega}}=\mathbb{C}_{BCH}=\mathbb{C}_{HT}=\{\vec{0}_{\ell}\}$ (\ie, $d(\mathbb{C}_{\overline{\Omega}}) = d(\mathbb{C}_{BCH}) = d(\mathbb{C}_{HT})=\infty$), due to their ease of computation. For each QT code listed, its eigenvalue set $\overline{\Omega}$ is given in terms of an index set $\mathcal{I}$, where $\overline{\Omega}=\{\xi^i : i \in \mathcal{I}\}$, for some primitive $m^{th}$ root of unity $\xi$.

\begin{table}[ht] 
\caption{Spectral bound versus ST and ZL bounds}
\setlength{\tabcolsep}{3pt}
\resizebox{0.485\textwidth}{!}{
\begin{tabular}{c|c|ccc|ccc|l}
No & $q$ & $\lambda$ & $m$ & $\ell$ & $d_{_{BCH}}$ & $d_{_{HT}}$ &  $d\se d_{_{SP}}$& \hspace{35pt}$\mathcal{I}$ for $\overline{\Omega}$ \\
\hline
1 & $2$ & $1$ & $23$ & $2$ & $5$ & $5$ & $7$ & $\{ 1\sk 4, 6, 8, 9, 12, 13, 16, 18\}$\\ 
2 & &  & $33$ & $2$  & $8$ & $10$ & $12$ & $\{0, 3, 5\sk 7, 9\sk 15, 18\sk 24, 26\sk 28, 30\}$\\
3 & &  & $39$ & $2$  & $7$ & $8$ & $12$ & $\{ 3, 6, 7, 9, 12\sk 15, 17\sk 19, 21, 23, 24, $\\&&&&&&&&$26\sk 31, 33\sk 38 \}$\\ 
4 & &  & $21$ & $3$  & $5$ & $6$ & $8$ & $\{ 3, 5\sk 7, 9, 10, 12\sk 15, 17\sk 20 \}$ \\ 
5 & &  & $33$ & $3$  & $5$ & $8$ & $11$ & $\{ 1\sk 4, 6, 8, 9, 11, 12, 15\sk 18, 21, 22,$\\&&&&&&&&$24, 25, 27, 29\sk 32 \}$\\ 
\hline
6 & $3$ & $1$ & $13$ & $2$ & $4$ & $5$ & $6$ & $\{ 0, 2, 4\sk 6, 10, 12 \}$ \\ 
7 & & & $20$ & $2$ & $5$ & $5$ & $8$ & $\{ 0, 1, 3\sk 5, 7\sk 10, 12, 15, 16 \}$\\ 
8 & &  & $40$ & $2$ & $11$ & $17$ & $20$ & $\{ 0, 2, 4\sk 8, 11\sk 19, 21\sk 26, 28, 29,$\\&&&&&&&&$ 31\sk 39 \}$\\
9 & &  & $26$ & $3$ & $5$ & $8$ & $10$ & $\{1\sk 4, 6, 8\sk 10, 12, 13, 17, 18, 20,$\\&&&&&&&&$23\sk 25 \}$\\
10 & &  & $44$ & $3$ & $10$ & $11$ & $18$ & $\{ 0\sk 7, 9\sk 13, 15\sk 23, 25, 27, 29\sk 31,$\\&&&&&&&&$33, 35\sk 37, 39, 41, 43 \}$\\
\hline
11 & $3$ & $-1$ & $20$ & $2$ & $4$ & $5$ & $6$ & $\{ 3, 6, 10\sk 12, 14, 15, 17\sk 19 \}$\\
12 & &  & $28$ & $2$ & $4$ & $6$ & $9$ & $\{ 0\sk 2, 4, 6, 7, 9, 11\sk 13, 17, 19, 22,$\\&&&&&&&&$ 24 \}$\\ 
13 & &  & $41$ & $2$ & $11$ & $13$ & $20$ & $\{ 0\sk 4, 6, 7, 9\sk 14, 17\sk 19, 21\sk 23,$\\&&&&&&&&$ 26\sk 31, 33, 34, 36\sk 40 \}$ \\
14 & &  & $28$ & $3$ & $3$ & $4$ & $6$ & $\{ 3, 10, 14, 15, 17, 18, 23, 24, 26, 27 \}$ \\ 
15 & &  & $28$ & $3$ & $7$ & $9$ & $11$ & $\{ 0\sk 2, 4\sk 9, 11\sk 13, 16, 17, 19\sk 22,$\\&&&&&&&&$24, 25 \} $ \\
\hline
\end{tabular}}\vspace{-10pt}
\label{tab:comparison}
\end{table}


\end{document}